\documentclass[sigconf]{acmart}

\newcommand{\myparatight}[1]{\smallskip\noindent{\bf {#1}:}~}

\usepackage{amsthm}
\usepackage{amsmath}

\usepackage{amssymb}

\usepackage{booktabs} 
\usepackage{array} 
\usepackage{multirow}
\usepackage{graphics}
\usepackage{graphicx}
\usepackage{algorithm}
\usepackage{algpseudocode}

\usepackage{amsfonts}
\usepackage{color, url}

\usepackage[skip=3pt]{caption}
\usepackage[skip=3pt]{subcaption}

\usepackage{tabto}
\usepackage{xcolor, colortbl}
\usepackage{enumitem}
\usepackage{bbm}
\usepackage{dsfont}

\usepackage{bm}
\usepackage[mathscr]{eucal}

\usepackage{hyperref}

\usepackage{makecell}

\usepackage{bigstrut,multirow,rotating}
\usepackage{placeins}
\usepackage{pgfplots}
\pgfplotsset{compat=1.17}
\makeatletter
\@namedef{r@tocindent4}{0pt}
\makeatother

\newtheorem{assumption}{Assumption}

\newtheorem{thm}{Theorem}

\DeclareMathOperator*{\argmax}{arg\,max}

\usepackage{xspace}

\newcommand{\algatt}{\textsf{BadUnlearn}\xspace}
\newcommand{\alg}{\textsf{UnlearnGuard}\xspace}
\newcommand{\algdis}{\textsf{UnlearnGuard-Dist}\xspace}
\newcommand{\algcos}{\textsf{UnlearnGuard-Dir}\xspace}

\copyrightyear{2025} 
\acmYear{2025} 
\setcopyright{cc}
\setcctype{by}
\acmConference[WWW Companion '25]{Companion Proceedings of the ACM Web Conference 2025}{April 28-May 2, 2025}{Sydney, NSW, Australia}
\acmBooktitle{Companion Proceedings of the ACM Web Conference 2025 (WWW Companion '25), April 28-May 2, 2025, Sydney, NSW, Australia}
\acmDOI{10.1145/3701716.3715494}
\acmISBN{979-8-4007-1331-6/25/04}

\makeatletter
\gdef\@copyrightpermission{
  \begin{minipage}{0.3\columnwidth}
   \href{https://creativecommons.org/licenses/by/4.0/}{\includegraphics[width=0.90\textwidth]{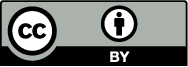}}
  \end{minipage}\hfill
  \begin{minipage}{0.7\columnwidth}
   \href{https://creativecommons.org/licenses/by/4.0/}{This work is licensed under a Creative Commons Attribution International 4.0 License.}
  \end{minipage}
  \vspace{5pt}
}
\makeatother

\begin{document}

\title{Poisoning Attacks and Defenses to Federated Unlearning}

\author{Wenbin Wang}
\authornote{Equal contribution. Wenbin Wang and Qiwen Ma conducted this research while they were interns under the supervision of Minghong Fang.}
\affiliation{
	\institution{ShanghaiTech University}
	\city{Shanghai}
	\country{China}
}

\author{Qiwen Ma}
\authornotemark[1]
\affiliation{
	\institution{Xidian University}
	\city{Xi'an}
	\country{China}
}

\author{Zifan Zhang}
\affiliation{
	\institution{North Carolina State University}
	\city{Raleigh}
	\country{USA}
}

\author{Yuchen Liu}
\affiliation{
	\institution{North Carolina State University}
	\city{Raleigh}
	\country{USA}
}

\author{Zhuqing Liu}
\affiliation{
	\institution{University of North Texas}
	\city{Denton}
	\country{USA}
}

\author{Minghong Fang}
\affiliation{
	\institution{University of Louisville}
	\city{Louisville}
	\country{USA}
}

\begin{abstract}
Federated learning allows multiple clients to collaboratively train a global model with the assistance of a server. However, its distributed nature makes it susceptible to poisoning attacks, where malicious clients can compromise the global model by sending harmful local model updates to the server. To unlearn an accurate global model from a poisoned one after identifying malicious clients, federated unlearning has been introduced. Yet, current research on federated unlearning has primarily concentrated on its effectiveness and efficiency, overlooking the security challenges it presents. In this work, we bridge the gap via proposing \algatt, the first poisoning attacks targeting federated unlearning. In \algatt, malicious clients send specifically designed local model updates to the server during the unlearning process, aiming to ensure that the resulting unlearned model remains poisoned. To mitigate these threats, we propose \alg, a robust federated unlearning framework that is provably robust against both existing poisoning attacks and our \algatt. The core concept of \alg is for the server to estimate the clients' local model updates during the unlearning process and employ a filtering strategy to verify the accuracy of these estimations. Theoretically, we prove that the model unlearned through \alg closely resembles one obtained by train-from-scratch. 
Empirically, we show that \algatt can effectively corrupt existing federated unlearning methods, while \alg remains secure against poisoning attacks.
\end{abstract}

\ccsdesc[500]{Security and privacy~Systems security}

\keywords{Federated Unlearning, Poisoning Attacks, Robustness}

\maketitle


\section{Introduction} \label{sec:intro}

Federated learning (FL)~\cite{McMahan17} allows clients to collaboratively train models while respecting privacy. However, its decentralized nature introduces vulnerabilities to poisoning attacks~\cite{fang2020local,shejwalkar2021manipulating,zhang2024poisoning,yin2024poisoning}. While Byzantine-robust aggregation methods~\cite{yin2018byzantine,guerraoui2018hidden,blanchard2017machine,fang2024byzantine,fang2022aflguard,fang2025FoundationFL} mitigate such attacks, they remain susceptible to sophisticated strategies. 
%
As a result, researchers have proposed various detection-based methods~\cite{cao2020fltrust,yueqifedredefense,nguyen2022flame} to identify malicious clients in the FL system. Despite their potential, these methods cannot fully prevent poisoning, as the global model may already be compromised.
To eliminate the impact of malicious clients on the final learned global model, one obvious approach is to retrain the model entirely from scratch after excluding malicious clients, a process often referred to as \emph{train-from-scratch}; however, this approach is computationally inefficient.
Recognizing these inefficiencies, federated unlearning (FU) mechanisms~\cite{tao2024communication,cao2023fedrecover,liu2024survey} have been designed to efficiently remove the influence of corrupted models.
While numerous studies have focused on enhancing the utility of the FU process, such as improving its effectiveness and efficiency, the security concerns associated with FU have been largely overlooked.

\myparatight{Our work}
In this work, we bridge this gap by introducing \algatt, the first poisoning attacks specifically designed to target FU.
Note that in current FU methods, the server requests clients to calculate exact local model updates and transmit these updates back during specific training rounds to achieve a more accurate global model. This communication between the server and clients during the FU process creates an opportunity for attackers to exploit. 
In our proposed approach, \algatt, malicious clients send meticulously crafted harmful local model updates to the server throughout the FU process. This manipulation aims to ensure that the final {\em unlearned model} remains similar to the poisoned model (referred to as the {\em learned model}) obtained during the FL phase. In other words, \algatt seeks to cause existing FU methods to fail by preventing them from successfully unlearning an accurate global model from the poisoned one.
We frame our attack as an optimization problem, where malicious updates depend on the server's aggregation rules. However, solving this is challenging due to non-differentiable aggregation methods like Median~\cite{yin2018byzantine} in FU techniques such as FedRecover~\cite{cao2023fedrecover}. To address this, we model each malicious update as a perturbed version of the learned model.

To counter the \algatt attack, we propose a secure FU framework, \alg. The core idea of \alg is that the server stores historical data, including clients' local model updates and global models, from the FL process. This data helps estimate local model updates during the FU phase. However, these estimations may not always be accurate, so we introduce a filtering strategy to verify them. The strategy ensures that an estimated update does not significantly deviate from the corresponding stored update. The first variant of our robust FU method, \algdis, considers an estimated update accurate if it is close in distance to the stored update. However, \algdis overlooks the direction of updates, allowing the attacker to manipulate the method with small magnitude updates. To address this, we introduce \algcos, where an estimated update is considered precise if it matches the direction of the original update from the FL phase.
We summarize our main contributions in this paper as follows:

\begin{list}{\labelitemi}{\leftmargin=1em \itemindent=-0.08em \itemsep=.2em}


\item 
We present \algatt attack, the first poisoning attacks targeting the FU process.

\item 
We introduce \alg, a FU framework that withstands both existing poisoning attacks and our BadUnlearn. We theoretically demonstrate that, under mild assumptions, the model unlearned by \alg closely approximates the one trained from scratch.

\item 
Extensive experiments across various poisoning attacks, aggregation rules, and FU methods demonstrate that \algatt effectively compromises existing FU methods, whereas our proposed \alg remains robust.

\end{list}


\section{Preliminaries and Threat Model} 

\label{sec:related}

\subsection{Federated Learning (FL)}
\label{fl_steps}

In FL, \(n\) clients collaboratively train a global model \(\mathbf{w}\) without sharing local dataset \(D_i\), where $D_i$ is client $i$'s training data, $i\ \in [n]$, and $[n]$ represent the set $\{1, 2, \dots, n\}$.
The model minimizes the loss \(\mathcal{L}(\mathbf{w}, D) = \sum_{i \in [n]} \frac{|D_i|}{|D|} \mathcal{L}_i(\mathbf{w}, D_i)\), where \(D = \cup_{i \in [n]} D_i\). At the training round $t$, the server broadcasts \(\mathbf{w}^t\) to all clients (Step I), clients train locally and compute updates \(\mathbf{g}_i^t\) (Step II), and the server aggregates updates using \(\mathsf{ARR}\) to update \(\mathbf{w}^{t+1} = \mathbf{w}^t - \eta \cdot \mathsf{ARR} \left\{\mathbf{g}_i^t: i \in [n]\right\}\) (Step III), where \(\mathsf{ARR}\) is the aggregation rule,  \(\eta\) is the learning rate.

\subsection{Federated Unlearning (FU)}
FU~\cite{tao2024communication,cao2023fedrecover,sheng2024robust} aims to mitigate the influence of malicious clients after a global model is trained, focusing on removing entire clients rather than individual data points, as in centralized learning. 
Existing research on FU has primarily focused on improving utility and accuracy, neglecting security challenges. FU seeks to eliminate the effects of compromised clients, which can introduce manipulated updates, but current techniques largely overlook the risks of poisoning attacks. To address these concerns, more research is needed to ensure FU methods are robust against malicious clients.

\subsection{Threat Model}

\myparatight{Attacker's goal and knowledge}We consider a scenario where an attacker controls malicious clients. During the FU process, these clients send crafted model updates to ensure the unlearned model remains similar to the model learned from the FL process. 
We consider three settings: \emph{full-knowledge}, where the attacker knows all clients' updates and the aggregation rule \(\mathsf{ARR}\); \emph{partial-knowledge}, where the attacker knows only the updates of malicious clients and the \(\mathsf{ARR}\); and \emph{black-box}, where the attacker knows only the updates of malicious clients and is unaware of the server's aggregation rule.

\myparatight{Defender's goal}Our goal is to develop an effective unlearning method with two key objectives: (i) robustness, ensuring resilience to poisoning attacks during the FU process, and (ii) aggregation rule independence, allowing compatibility with any server aggregation rule. Moreover, the method assumes the server remains unaware of both the attacks and the number of malicious clients.


\section{Our \algatt} 
\label{BadUnlearn_sec}

In existing FU methods like FedRecover, the server collects and aggregates clients' local model updates during specific rounds. However, this communication allows malicious clients to introduce harmful updates, which can undermine the FU process. In our proposed \algatt attack, the attacker strategically craft updates to ensure the resulting unlearned model remains close to the final learned model from the FL phase. This attack can be formulated as an optimization problem:
\begin{align}
\min \Vert ^{\#}\mathbf{w} - {\mathsf{ARR}}\{\mathbf{g}_i^t: i \in [n]\} \Vert,
\end{align}
where \(^{\#}\mathbf{w}\) is the final learned model from the FL phase, 
${\mathsf{ARR}}\{\mathbf{g}_i^t: i \in [n]\}$ denotes the aggregated model update after the attack during the FU process, $\left\| \cdot  \right\|$ denotes the $\ell_2$ norm.

Solving this is challenging due to the non-differentiability of most aggregation rules. To overcome this, we approximate the solution by assuming each malicious update is a perturbed version of the pre-attack aggregated model \(^{\#}\mathbf{w}\). 
Specifically,
$
\mathbf{g}_k^t = {^{\#}\mathbf{w}} + \epsilon  \psi,  k \in \mathcal{B},
$
where \(\epsilon\) is a scaling factor, and \(\psi\) is a perturbation vector,
$\mathcal{B}$ represent the set of malicious clients. 
We can define \(\psi\) as a predefined parameter, for example, \(\psi = -\text{sign}({^{\#}\mathbf{w}})\), where \(\text{sign}\) represents the coordinate-wise sign of \({^{\#}\mathbf{w}}\), thereby simplifying the optimization problem to:
\begin{equation}
\label{recover_attack_convert2}
\begin{split}
\argmax_{\epsilon}  & - \Vert {^{\#}\mathbf{w}} - {\mathsf{ARR}}\{\mathbf{g}_i^t: i \in [n]\} \Vert  \\
\text{s.t.} \quad  \mathbf{g}_k^t &= {^{\#}\mathbf{w}} + \epsilon  \psi,
\quad \psi = -\text{sign}({^{\#}\mathbf{w}}), \quad k \in \mathcal{B}.
\end{split}
\end{equation}

The attacker determines \(\epsilon\) heuristically, adjusting it based on changes in the objective function. Once \(\epsilon\) is set, malicious updates are sent to the server, disrupting the FU process.


\section{\alg} 
\label{recover method}

\subsection{Overview}

In our proposed \alg, the server retains historical data from the FL process. After excluding malicious clients, it uses FU to unlearn and rebuild a clean global model. 
Starting with a new initialization, the server iteratively trains the model as outlined in Section~\ref{fl_steps}, estimating client updates from stored data and requesting the clients to compute the exact updates when necessary.

\subsection{Model Updates Prediction}
\label{model_prediction}

In our \alg, the server uses historical data from the FL process to estimate client updates during FU. Let $\tilde{\mathbf{g}}_i^t$ and $\tilde{\mathbf{w}}^t$ denote the local and global model updates at round $t$ during FL, stored by the server. 
In FU, ${\mathbf{g}}_i^t$ denotes the exact update (calculated using local training data), while $\bar{\mathbf{g}}_i^t$ represents the estimated update for client $i$.
Using the Cauchy mean value theorem, $\bar{\mathbf{g}}_i^t$ is given by:
\begin{align}
\bar{\mathbf{g}}_i^t = \tilde{\mathbf{g}}_i^t + \mathbf{H}_{i}^{t} ({\mathbf{w}}^{t} - \tilde{\mathbf{w}}^{t}),
\end{align}
where $\mathbf{H}_{i}^{t}$, the Hessian matrix for client $i$, is calculated as $\mathbf{H}_{i}^{t} = \intop_{0}^{1} \mathbf{H}(\tilde{\mathbf{w}}^{t} + y ({\mathbf{w}}^{t} - \tilde{\mathbf{w}}^{t})) dy$. Direct computation of $\mathbf{H}_{i}^{t}$ is computationally expensive; hence, the L-BFGS algorithm approximates it using limited historical data. 
Define the global model differences as $\Delta {\mathbf{w}}^{t} = {\mathbf{w}}^{t} - \tilde{\mathbf{w}}^{t}$, and the local update differences for client $i$ at training round $t$ as $\Delta {\mathbf{g}}_i^{t} = \bar{\mathbf{g}}_i^t - \tilde{\mathbf{g}}_i^t$. For the last $s$ rounds, let $\Delta {\mathbf{W}}^{t,s} = \{ \Delta {\mathbf{w}}^{t-s}, \dots, \Delta {\mathbf{w}}^{t-1} \}$ and $\Delta {\mathbf{G}}_i^{t,s} = \{ \Delta {\mathbf{g}}_i^{t-s}, \dots, \Delta {\mathbf{g}}_i^{t-1} \}$. L-BFGS uses $\Delta {\mathbf{W}}^{t,s}$, $\Delta {\mathbf{G}}_i^{t,s}$, and $\Delta {\mathbf{w}}^{t}$ to compute Hessian-vector products $\mathbf{H}_{i}^{t} ({\mathbf{w}}^{t} - \tilde{\mathbf{w}}^{t})$, which are then used to predict $\bar{\mathbf{g}}_i^t$.

\subsection{\algdis}
\label{first_method}

In Section~\ref{model_prediction}, we demonstrate that the server can estimate client updates during FU using historical data, though inaccuracies may arise and degrade model quality over time. To address this, we propose \algdis, a distance-based calibration technique to minimize estimation errors. 
The central idea of \algdis is that, during a specific training round, if a client's local model update estimated in the FU process significantly deviates from the corresponding local model update stored during the FL process, it indicates an inaccurate estimation.
Specifically, we consider the estimated update $\bar{\mathbf{g}}_i^t $ accurate if it satisfies the following condition:
\begin{align}
\label{first_check}
\max_{t_1 \in [t - r, t]} \left\| \bar{\mathbf{g}}_i^t - \tilde{\mathbf{g}}_i^{t_1} \right\| \le 
\max_{t_2, t_3 \in [t - r, t]} \left\| \tilde{\mathbf{g}}_i^{t_2} - \tilde{\mathbf{g}}_i^{t_3} \right\|,
\end{align}
where \(r\) is the buffer parameter. The left term measures the largest deviation of \(\bar{\mathbf{g}}_i^t\) from the stored updates, while the right term represents the largest deviation among the stored updates themselves. This ensures that \(\bar{\mathbf{g}}_i^t\) aligns with the client's historical update behavior. If this condition is met, the server uses \(\bar{\mathbf{g}}_i^t\) to compute the global model for round \(t\). Otherwise, it requests the exact update \(\mathbf{g}_i^t\) from the client to ensure reliability.

\subsection{\algcos}
\label{second_method}

Our proposed \algdis uses distance-based calibration to ensure the accuracy of estimated client updates, focusing on their magnitude. However, \algdis overlooks update directions, which can result in suboptimal unlearning even when the distance condition in Eq.~(\ref{first_check}) is satisfied. To address this, we introduce \algcos, a direction-aware calibration method that considers both magnitude and direction to reduce estimation errors during FU. In \algcos, an estimated update \(\bar{\mathbf{g}}_i^t\) is deemed accurate if it satisfies:
\begin{align}
\label{second_check}
\max\limits_{t_1 \in [t-r, t]} \text{cos}(\bar{\mathbf{g}}_{i}^{t}, \tilde{\mathbf{g}}_{i}^{t_1}) \le 
\max\limits_{t_2, t_3  \in [t-r, t]} \text{cos}(\tilde{\mathbf{g}}_{i}^{t_2}, \tilde{\mathbf{g}}_{i}^{t_3}),
\end{align}
where \(\text{cos}(\cdot, \cdot)\) measures cosine similarity. This ensures \(\bar{\mathbf{g}}_i^t\) aligns with the trend of historical updates without deviation. If the condition holds, we rescale \(\bar{\mathbf{g}}_i^t\) to match the typical magnitude of past updates using
$
\bar{\mathbf{g}}_i^t \leftarrow \frac{\bar{\mathbf{g}}_i^t}{\|\bar{\mathbf{g}}_i^t\|} \times \text{Median}\{\|\tilde{\mathbf{g}}_{i}^{t-r}\|, \dots, \|\tilde{\mathbf{g}}_{i}^{t}\|\},
$
where the median is computed over the magnitudes of updates from rounds \(t-r\) to \(t\). Algorithm~\ref{our_alg_app} provides the pseudocode of \alg, assuming \(n\) clients, with \(m\) malicious clients removed after detection, leaving \(n-m\) for FU.

\begin{algorithm}[t] 
	\caption{The \alg method.}\label{alg:unifiedfedrecoverydetailed}
	\begin{algorithmic}[1]
		\renewcommand{\algorithmicrequire}{\textbf{Input:}}
		\renewcommand{\algorithmicensure}{\textbf{Output:}}
		\Require Set of clients in the FU phase $\mathcal{S}=\{i \mid m+1\leq i \leq n \}$; global models and model updates information in FL phase $ \tilde{\mathbf{w}}^0,  \tilde{\mathbf{w}}^1, \ldots, \tilde{\mathbf{w}}^T;  \tilde{\mathbf{g}}^0_i, \tilde{\mathbf{g}}^1_i,\ldots, \tilde{\mathbf{g}}^{T-1}_i$; learning rate $\eta$; L-BFGS buffer size $s$; aggregation rule 
 $\mathsf{ARR}$. 
		\Ensure Unlearned model $\mathbf
  {w}^{T}$. 
		\State Initialize $\mathbf{{w}}_0 \gets  \tilde{\mathbf{w}}^0$.
		\For{$t = 0 \text{ to } T-1$} 
            \If{$t < r $}
                \State The server broadcasts $\mathbf{{w}}_t$ to the  clients in $\mathcal{S}$.
                \State // Each client trains its local model and computes model update
                \State Each client $i$ computes 
                $\mathbf{g}^{t}_i=\nabla \mathcal{L}_i(\mathbf{{w}}^t)$.
                \State // Aggregation and global model updating.
                \State $\mathbf{{w}}^{t+1} \gets \mathbf{{w}}^{t} - \eta\cdot\mathsf{ARR}(\mathbf{g}^t_{m+1}, \ldots, \mathbf{g}^t_{n})$. 
            \Else
                \State Update L-BFGS buffers 
                $\Delta {\mathbf{W}}^{t,s}$, $\Delta{\mathbf{G}}_i^{t,s}$.
                    \For{$i = m+1 \text{ to } n$}
                        \State  $\Delta\mathbf{w}^t = \mathbf{w}^t - \tilde{\mathbf{w}}^t$ 

                        \State // We denote $\text{diag}(\mathbf{X})$ as the diagonal matrix of $\mathbf{X}$, and  $\text{tril}(\mathbf{X})$ the lower triangular matrix of $\mathbf{X}$.
                        
                        \State  $\mathbf{A}^{t,s}_i = (\Delta\mathbf{W}^{t,s})^{\top}\Delta\mathbf{G}_i^{t,s}$, $\mathbf{D}^{t,s}_i = \text{diag}(\mathbf{A}^{t,s}_i)$

                        \State$\mathbf{L}^{t,s}_i = \text{tril}(\mathbf{A}^{t,s}_i)$ 
                        
                        \State  $\sigma = ((\Delta \mathbf{g}^{t-2}_i)^{\top} \Delta \mathbf{w}^{t-2}) / ((\Delta \mathbf{w}^{t-2})^{\top}\Delta \mathbf{w}^{t-2})$
                        \State $\mathbf{p} = \begin{bmatrix}
                        -\mathbf{D}_i^{t,s} & (\mathbf{L}_i^{t,s})^{\top} \\
                        \mathbf{L}_i^{t,s} & \sigma(\Delta\mathbf{W}^{t,s})^{\top}\Delta\mathbf{W}^{t,s}
                        \end{bmatrix}^{-1}
                        \begin{bmatrix}
                        (\Delta\mathbf{G}_i^{t,s})^{\top}\Delta\mathbf{w}^t \\
                        \sigma(\Delta\mathbf{W}^{t,s})^{\top}\Delta\mathbf{w}^t 
                        \end{bmatrix}$ 
                        \State $\mathbf{H}_i^t\Delta\mathbf{w}^t = \sigma \Delta\mathbf{w}^t - \begin{bmatrix}
                        \Delta\mathbf{G}^{t,s}_i &
                        \sigma(\Delta\mathbf{W}^{t,s})
                        \end{bmatrix}\mathbf{p}$
                        \State $\bar{\mathbf{{g}}}^t_i= \tilde{\mathbf{g}}_i^t + \mathbf{H}_i^t\Delta\mathbf{w}^t$
                        \State // Exact training
                        \If{condition in Eq.~(\ref{first_check}) or Eq.~(\ref{second_check}) is not met}
                            \State Compute $\mathbf{g}^t_i=\nabla \mathcal{L}_i(\mathbf{w}^t)$, and update $\bar{\mathbf{g}}^t_i \gets \mathbf{g}^t_i$ (rescale $\bar{\mathbf{g}}_i^t$ if necessary).
                        \EndIf
                    \EndFor
                    \State $\mathbf{{w}}^{t+1} \gets \mathbf{{w}}^{t} - \eta\cdot\mathsf{ARR}(\mathbf{\bar{g}}^t_{m+1}, \ldots, \mathbf{\bar{g}}^t_n)$
                \EndIf
		\EndFor
	\end{algorithmic}
 \label{our_alg_app}
\end{algorithm}

\subsection{Security Analysis}

\begin{assumption} \label{assumption1}
The loss function $\mathcal{L}_i$ for client $i$ is $\mu$-strongly convexity and $L$-smoothness. Specifically, the following inequalities hold for any vectors $\mathbf{w}_1, \mathbf{w}_2 \in \mathbb{R}^d$:
\begin{align}
    \mathcal{L}_i(\mathbf{w}_1) \ge \mathcal{L}_i(\mathbf{w}_2) + \nabla\mathcal{L}_i(\mathbf{w}_2)^\top (\mathbf{w}_1 - \mathbf{w}_2) + \frac{\mu}{2} \|\mathbf{w}_1 - \mathbf{w}_2  \|^2, \nonumber \\
    \mathcal{L}_i(\mathbf{w}_1) \le \mathcal{L}_i(\mathbf{w}_2) + \nabla\mathcal{L}_i(\mathbf{w}_2)^\top (\mathbf{w}_1 - \mathbf{w}_2) + \frac{L}{2} \|\mathbf{w}_1 - \mathbf{w}_2  \|^2 \nonumber.
 \end{align}
\end{assumption}

\begin{assumption} \label{assumption2}
The L-BFGS algorithm bounds the error in approximating Hessian-vector products as
$
\Vert\mathbf{H}_{i}^{t} ({\mathbf{w}}^{t} - \tilde{\mathbf{w}}^{t}) + \tilde{\mathbf{g}}_i^t - \bar{\mathbf{g}}_i^t\Vert \le M,  \text{for any $i$ and $t$},
$
where \(M\) is a finite positive constant.
    \label{as:approx}
\end{assumption}

\begin{thm}  
\label{theorem_1}
Under Assumptions~\ref{assumption1} and \ref{assumption2}, with a learning rate \(\eta \le \min(\frac{1}{\mu}, \frac{1}{L})\) and effective detection of all malicious clients during FL, the discrepancy between the global model unlearned via \algdis (or \algcos) and the train-from-scratch model in each round \(t > 0\) is bounded as:
\begin{align}
\|\mathbf{w}^{t}-\ddot{\mathbf{w}}^{t}\|  \leq(\sqrt{1-\eta\mu})^{t}\|{\mathbf{w}}^{0}-{\ddot{\mathbf{w}}}^{0}\|+\frac{1-(\sqrt{1-\eta\mu})^{t}}{1-\sqrt{1-\eta\mu}}\eta M  \nonumber 
\end{align}
where \({\mathbf{w}}^t\) and \(\ddot{\mathbf{w}}^t\) are the global models from \algdis (or \algcos) and train-from-scratch, respectively, in round \(t\).
\end{thm}
\begin{proof}
See Appendix~\ref{sec:appendix_1} for the proof.
\end{proof}

\section{Experiments} \label{sec:exp}

\subsection{Experimental Setup}

Due to space constraints, we only use the MNIST dataset.
For aggregation during FL training and unlearning, we mainly use FedAvg~\cite{McMahan17}, Median~\cite{yin2018byzantine}, and Trimmed-mean (TrMean)~\cite{yin2018byzantine}, while also evaluating Bulyan~\cite{guerraoui2018hidden}, Krum~\cite{blanchard2017machine}, and FLAME~\cite{nguyen2022flame}. We examine three attack schemes: Trim attack~\cite{fang2020local}, which disrupts FL by crafting malicious updates to distort global models; Backdoor attack~\cite{bagdasaryan2020backdoor}, where triggers embedded in training data lead to targeted malicious outcomes; and our \algatt, aimed to manipulate FU process.
Full-knowledge attack setting is considered by default.

\subsubsection{FL and FU Settings}

Training involved 2,000 rounds for MNIST (learning rate $3 \times 10^{-4}$, batch size 32). 
A CNN architecture consisting of two convolutional layers, each paired with a pooling layer, followed by two fully connected layers.
We assume 100 clients, with 20\% being malicious. FU follows FL settings and uses 80 clients, with 20\% performing attacks like Trim attack or \algatt attack. Malicious clients detected in FL are removed during FU.  
The buffer parameter is $r=5$, assuming Trim attack in FL and BadUnlearn attack in FU.
For MNIST, we simulate the Non-IID distribution using the approach in~\cite{fang2020local}. 
The Non-IID degree is set to 0.5 following~\cite{fang2020local}.


\begin{table}[t]
  \centering
   \scriptsize
  \caption{Results of various aggregation rules under attacks where the attacker manipulates the FL process on MNIST.
  }
    \label{attack_fl}
    \begin{tabular}{|c|c|c|c|c|}
    \hline
    Dataset & Attack FL & FedAvg & Median  & TrMean \\
   \hline
    \multirow{3}{*}{MNIST} & No attack &  0.04     &   0.06    & 0.06 \\
\cline{2-5}          & Trim attack &   0.23    &   0.41    & 0.20 \\
\cline{2-5}          & Backdoor attack &   0.02 / 0.99    &  0.09 / 0.01     & 0.06 / 0.01 \\
\hline
    \end{tabular}%
	\vspace{-.15in}
\end{table}%

\begin{table}[t]
  \centering
 \scriptsize
  \caption{Results of Train-from-scratch and Historical methods on MNIST dataset, where the attacker manipulates the FL process, and the server aims to obtain a clean global model.}
    \label{scratch_history_result}
    \begin{tabular}{|c|c|c|c|}
    \hline
    Attack FL  & $\mathsf{ARR}$ & Train-from-scratch & Historical \\
    \hline
    \multirow{3}{*}{Trim attack} & FedAvg  &   0.05    & 0.90 \\
\cline{2-4}          & Median  &  0.07     & 0.90 \\
\cline{2-4}          & TrMean &    0.06   & 0.90 \\
   \hline
    \multirow{3}{*}{Backdoor attack} & FedAvg  &   0.05 / 0.00    & 0.85 / 0.00 \\
\cline{2-4}          & Median  &  0.06 / 0.00     & 0.89 / 0.00 \\
\cline{2-4}          & TrMean &  0.06 / 0.00     & 0.86 / 0.00 \\
   \hline
    \end{tabular}%
\vspace{-.15in}
\end{table}%

\begin{table}[t]
	\centering
	 \scriptsize
	 \addtolength{\tabcolsep}{-1.965pt}
	\caption{Results of FU methods on the MNIST dataset, where the attacker manipulates both FL and FU processes.}
	\label{main_mnist}
	\begin{tabular}{|c|c|c|c|c|c|}
		\hline
		Attack FL  & Attack FU & $\mathsf{ARR}$ & FedRecover & \algdis & \algcos \\
		\hline
		\multirow{9}{*}{No attack} & \multirow{3}{*}{No attack} & FedAvg  &  0.06   &   0.05   & 0.05 \\
		\cline{3-6}          &       & Median  &   0.05   &    0.06   & 0.04 \\
		\cline{3-6}          &       & TrMean &  0.07   &    0.07   & 0.04 \\
		\cline{2-6}          & \multirow{3}{*}{Trim attack} & FedAvg  &   0.05   &  0.06  & 0.06 \\
		\cline{3-6}          &       & Median  &   0.12   &   0.12   & 0.08 \\
		\cline{3-6}          &       & TrMean &   0.15   &   0.11   & 0.06 \\
		\cline{2-6}          & \multirow{3}{*}{\algatt} & FedAvg  &   0.06   &   0.05   & 0.05 \\
		\cline{3-6}          &       & Median  &   0.07    &   0.11    & 0.07 \\
		\cline{3-6}          &       & TrMean &    0.06   &   0.13   & 0.05 \\
		\hline
		\multirow{9}{*}{Trim attack} & \multirow{3}{*}{No attack} & FedAvg  &  0.05     &   0.06    & 0.05 \\
		\cline{3-6}          &       & Median  &  0.11     &   0.14    & 0.09 \\
		\cline{3-6}          &       & TrMean &    0.15   &    0.12   & 0.06 \\
		\cline{2-6}          & \multirow{3}{*}{Trim attack} & FedAvg  &  0.06     &    0.07   & 0.06 \\
		\cline{3-6}          &       & Median  &   0.15    &   0.14    & 0.10 \\
		\cline{3-6}          &       & TrMean &    0.16   &   0.15    & 0.07 \\
		\cline{2-6}          & \multirow{3}{*}{\algatt} & FedAvg  &  0.24     &    0.07   & 0.06 \\
		\cline{3-6}          &       & Median  &   0.39    &    0.14   & 0.09 \\
		\cline{3-6}          &       & TrMean &   0.23    &   0.14    & 0.06 \\
		\hline
		\multirow{9}{*}{\makecell {Backdoor \\ attack}} & \multirow{3}{*}{No attack} & FedAvg  &   0.05 / 0.00    &    0.02 / 0.00   & 0.02 / 0.00  \\
		\cline{3-6}          &       & Median  &   0.12 / 0.01    &    0.11 / 0.01   & 0.10 / 0.00 \\
		\cline{3-6}          &       & TrMean &   0.07 / 0.00    &   0.06 / 0.00    & 0.06 / 0.00 \\
		\cline{2-6}          & \multirow{3}{*}{Trim attack} & FedAvg  &   0.06 / 0.01    &   0.03 / 0.01    & 0.03 / 0.00 \\
		\cline{3-6}          &       & Median  &   0.12 / 0.00    &  0.12 / 0.00     & 0.11 / 0.00 \\
		\cline{3-6}          &       & TrMean &  0.04 / 0.00     &   0.06 / 0.00    & 0.05 / 0.00 \\
		\cline{2-6}          & \multirow{3}{*}{\algatt} & FedAvg  &   0.04 / 0.98    &   0.03 / 0.01    & 0.03 / 0.00 \\
		\cline{3-6}          &       & Median  &   0.10 / 0.02    &  0.13 / 0.00     & 0.12 / 0.00 \\
		\cline{3-6}          &       & TrMean &   0.08 / 0.03    &     0.07 / 0.00  & 0.07 / 0.00 \\
		\hline
	\end{tabular}%
	\vspace{-.10in}
\end{table}%

\subsubsection{Evaluation Metrics} 
We consider two metrics: testing error rate (TER) and attack success rate (ASR). Higher TER and ASR indicate stronger attacks, lower values reflect better defenses.

\subsubsection{Compared Methods}We compare \alg{} with three unlearning methods: Train-from-scratch, Historical-information-only (Historical)~\cite{cao2023fedrecover}, and FedRecover~\cite{cao2023fedrecover}. Historical method rebuilds the global model using stored updates from the FL process, refining it incrementally without communication with client during FU.

\begin{table}[t]
  \centering
   \addtolength{\tabcolsep}{-3.55pt}
   \scriptsize
  \caption{Partial-knowledge and black-box attacks.}
    \begin{tabular}{|c|c|c|c|c|c|c|}
     \hline
    \multirow{3}{*}{ARR} & \multicolumn{3}{c|}{Partial-knowledge} & \multicolumn{3}{c|}{Black-box} \\
\cline{2-7}          & FedRecover & {\makecell {UnlearnGuard \\ -Dist}}   & {\makecell {UnlearnGuard \\ -Dir}}    & FedRecover &  {\makecell {UnlearnGuard \\ -Dist}}    & {\makecell {UnlearnGuard \\ -Dir}}   \\
     \hline
    FedAvg &   0.23    &  0.07     &  0.06     &   0.22    &  0.06     & 0.05 \\
     \hline
    Median &  0.38     &   0.14    &    0.09   &  0.38     & 0.14      & 0.08 \\
     \hline
    TrMean &  0.23     &   0.14    &   0.06    &  0.23     &  0.13     &  0.05 \\
    \hline
    \end{tabular}%
  \label{partial_black_attack}%
  \vspace{-.10in}
\end{table}%

\begin{table}[t]
    \centering
      \scriptsize
    \caption{The server employs advanced rules for FL and FU.}
    \label{more_agg}
    \begin{tabular}{|c|c|c|c|c|}
       \hline
        $\mathsf{ARR}$ & Learned model  \(^{\#}\mathbf{w}\) & FedRecover & \algdis & \algcos \\
       \hline
        Bulyan  & 0.13 & 0.19 & 0.12 & 0.13 \\
         \hline
        Krum & 0.10 & 0.16 & 0.11 & 0.12 \\
         \hline
        FLAME& 0.06 & 0.21 & 0.09 & 0.08 \\
       \hline
    \end{tabular}
	\vspace{-.10in}
\end{table}

\begin{table}[t]
    \centering
      \scriptsize
         \addtolength{\tabcolsep}{-0.55pt}
    \caption{Malicious clients execute advanced attacks in FL.}
    \label{more_attack}
    \begin{tabular}{|c|c|c|c|c|}
       \hline
        Attack  & Learned model  \(^{\#}\mathbf{w}\)& FedRecover & \algdis & \algcos \\
        \hline
Min-Max & 0.22 & 0.20 & 0.13 & 0.12 \\
\hline
Min-Sum & 0.23 & 0.24 & 0.12 & 0.09 \\
\hline
LIE& 0.18 & 0.19 & 0.09 & 0.08 \\
\hline
    \end{tabular}
	\vspace{-.12in}
\end{table}

\subsection{Experimental Results}

\myparatight{\algatt compromises FU methods, while \alg remains robust}Table~\ref{attack_fl} presents the performance of aggregation strategies under attacks during FL on the MNIST dataset.
``Attack FL'' means malicious clients perform attack during FL.
Backdoor attack results are presented as ``TER / ASR'', while only TER is reported for other attacks.
``No attack'' means all clients are benign. Table~\ref{scratch_history_result} illustrates the results of Train-from-scratch and Historical methods.
We observe from Table~\ref{attack_fl} that aggregation strategies like FedAvg, Median, and TrMean are vulnerable to poisoning attacks, while Table~\ref{scratch_history_result} indicates the Historical method fails to recover a clean global model when FL is compromised.

Table~\ref{main_mnist} presents the results of the FedRecover, \algdis, and \algcos methods on the MNIST dataset.
``Attack FU'' means malicious clients perform attack during FU.
Based on Tables~\ref{attack_fl}-\ref{main_mnist}, we can draw several observations.
\algatt effectively compromises FedRecover, a leading FU technique. Additionally, FedRecover demonstrates resilience against Trim attacks during the FU phase, consistent with findings in~\cite{cao2023fedrecover}. Furthermore, our \algdis and \algcos, exhibit robustness against various poisoning attacks and successfully unlearn a clean global model from a compromised one.

\myparatight{Partial-knowledge and black-box attacks}Table~\ref{partial_black_attack} presents the results of partial-knowledge and black-box attacks. In the black-box setting, the attacker is unaware of the server's aggregation rule and uses the Median rule to generate malicious updates. The results demonstrate that even in practical scenarios like the black-box setting, \algatt remains effective in compromising FedRecover.

\myparatight{Complex aggregation rules and advanced attacks}
We examines scenarios with advanced aggregation methods like Bulyan~\cite{guerraoui2018hidden}, Krum~\cite{blanchard2017machine}, and FLAME~\cite{nguyen2022flame}, where the attacker applies Trim attack during FL and \algatt attack during FU. The same aggregation method is consistently applied across both FL and FU. 
Note that FLAME is an advanced method that can detect malicious clients.
As shown in Table~\ref{more_agg}, our unlearning methods effectively mitigate poisoning attacks, even under these sophisticated aggregation strategies. We also evaluate cases where clients execute advanced attacks during FL, such as Min-Max~\cite{shejwalkar2021manipulating}, Min-Sum~\cite{shejwalkar2021manipulating}, and LIE~\cite{baruch2019little}, while continuing to employ \algatt during FU. Results in Table~\ref{more_attack} show that FedRecover remains vulnerable to \algatt.

\myparatight{Adaptive attack}
We evaluate \alg against a challenging Adaptive attack~\cite{shejwalkar2021manipulating}, where the attacker has full knowledge of the FU process, including clients' updates, the server's aggregation rule, and filtering strategies (e.g., Eq.~(\ref{first_check}) or Eq.~(\ref{second_check})). Malicious clients apply the Trim attack during FL and the Adaptive attack during FU. As shown in Table~\ref{adaptive_attack}, both \algdis and \algcos demonstrate resilience to these strong attacks.

\myparatight{Attack more complex FU methods}
We examine a scenario where the server uses the advanced FU method,
such as FATS~\cite{tao2024communication} or TFUA~\cite{sheng2024robust}. 
The attacker performs Trim attack during FL, and our \algatt attack during FU. 
The testing error rates for FATS and TFUA are 0.25 and 0.22, respectively, demonstrating that \algatt effectively disrupts FATS and TFUA during the FU process.

\myparatight{Server storage overhead}Our \alg stores historical data with a space complexity of \(O(d)\), requiring \(O(dnT)\) storage for \(d\) parameters, \(n\) clients, and \(T\) training rounds. In our experiments, it used up to 200 GB, manageable for modern servers.

\myparatight{Ablation studies}Fig.~\ref{fig:Effect_of_degree_of_non_iid(attack)} illustrates the impact of Non-IID degree, with larger x-axis values indicating more Non-IID data (default: 0.5). Both \algdis and \algcos show similar increasing trends, but \algcos maintains a consistently lower TER, while FedRecover exhibits the highest TER, worsening with increasing Non-IID.
Fig.~\ref{fig:Percentage_of_Malicious_Clients} shows that as malicious clients increase from 5\% to 50\%, FedRecover's TER rises from 0.12 to 0.89, while \alg remain robust.
Fig.~\ref{fig:Effect_of_number_of_p} shows the impact of buffer size \(r\) on \alg, with TER remaining generally stable.

\myparatight{Discussions}Our attack differs from traditional FL attacks by targeting the FU phase to prevent effective unlearning, ensuring the unlearned model mirrors the learned one regardless of testing error. Unlike FL poisoning attacks that corrupt training, we aim to compromise unlearning. 
Furthermore, our attack is feasible. 
Malicious clients remain dormant during FL phase—they do not perform attacks and act like benign clients—thereby evading detection during FL. These clients then execute their attack during FU process.
Note that the server can detect malicious clients during FU process.

\begin{table}[t]
    \centering
     \scriptsize
      \caption{Results of Adaptive attack.}
    \begin{tabular}{|c|c|c|}
       \hline
        $\mathsf{ARR}$ & \algdis & \algcos \\
        \hline
        FedAvg & 0.06 & 0.05 \\
        \hline
        Median & 0.16 & 0.08 \\
        \hline
       TrMean & 0.15 & 0.07 \\
       \hline
    \end{tabular}
  \label{adaptive_attack}
\vspace{-.18in}
\end{table}

\begin{figure}[!t]
	\centering
\subfloat[Degree of Non-IID.]{\includegraphics[width=0.16 \textwidth]{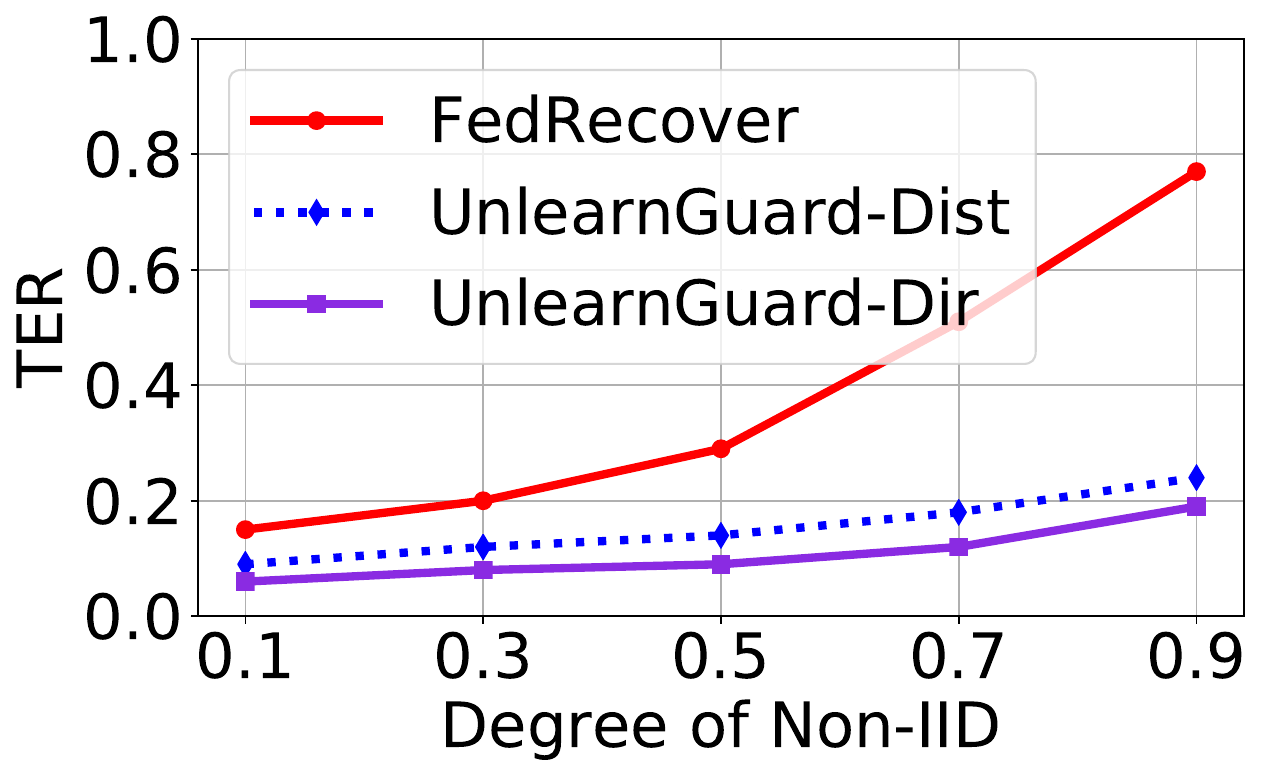}\label{fig:Effect_of_degree_of_non_iid(attack)}}
	\subfloat[Malicious fraction.]{\includegraphics[width=0.16 \textwidth]{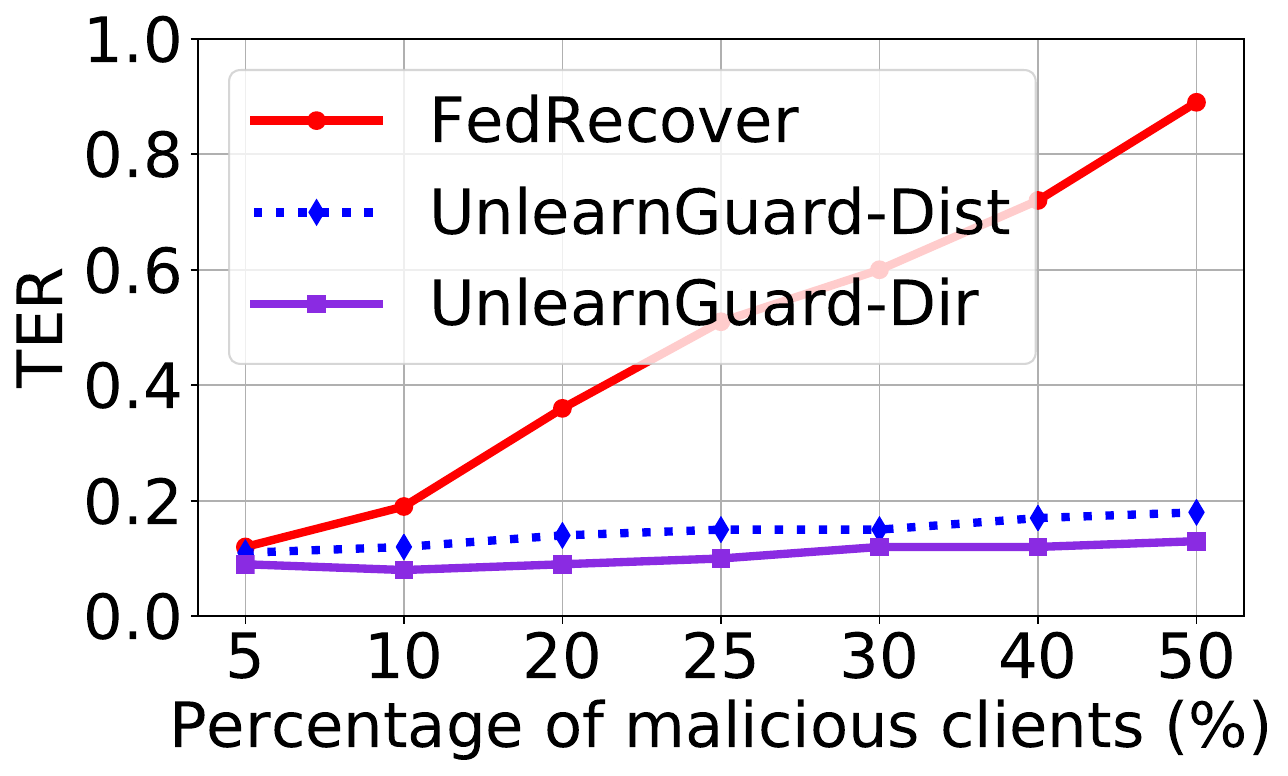}\label{fig:Percentage_of_Malicious_Clients}} 
	\subfloat[Number of $r$.]{\includegraphics[width=0.16 \textwidth]{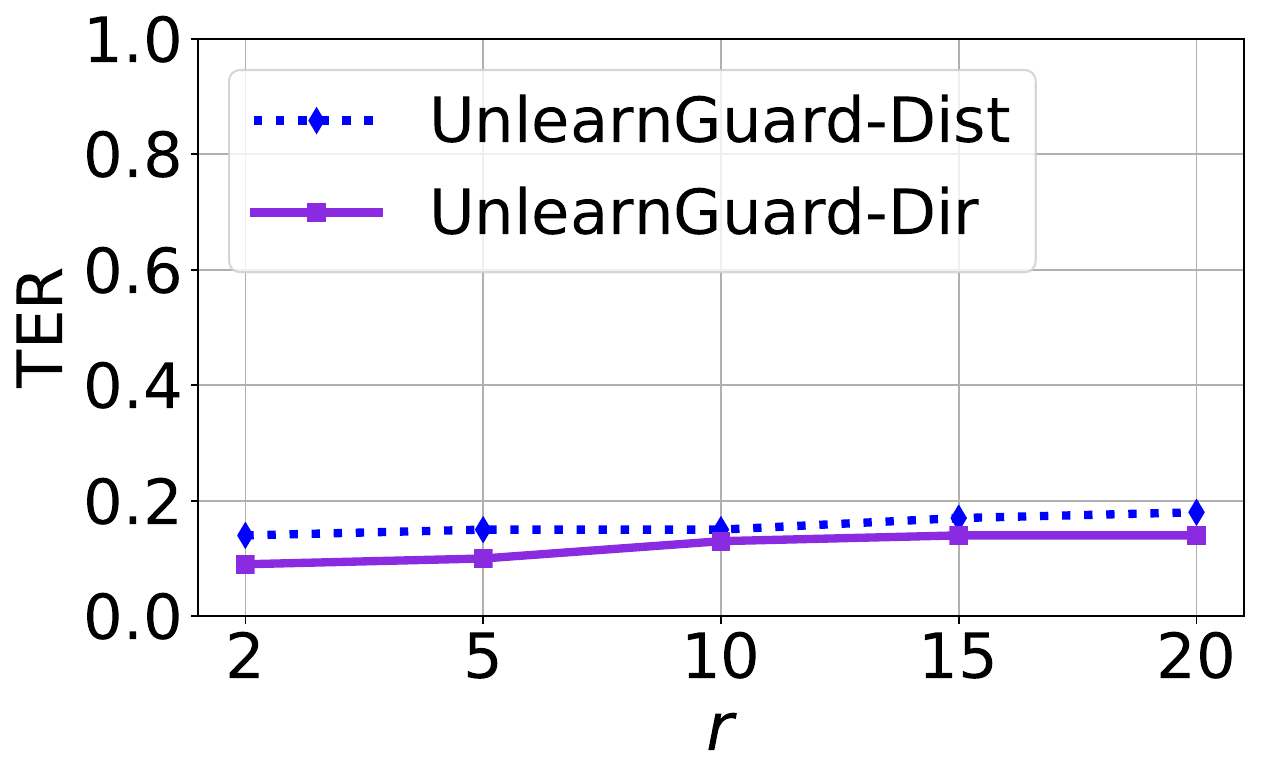}\label{fig:Effect_of_number_of_p}} 
 \caption{Impact of degree of Non-IID, percentage of malicious clients during FU, and the number of $r$ on MNIST.
}
 \label{cost_fig}
 \vspace{-.15in}
\end{figure}




%
%

\section{Conclusion}
\label{sec:conclusion}

We proposed \algatt, a novel poisoning attack targeting FU, which evades detection and disrupts unlearning by subtly altering model updates. 
To counter this, we introduced \algdis and \algcos, robust unlearning methods proven secure against poisoning attacks both theoretically and empirically.

\bibliographystyle{ACM-Reference-Format}
\bibliography{refs}


\appendix

\section{Proof of Theorem~\ref{theorem_1}} 
\label{sec:appendix_1}

Here, we establish the bound of $\|{\mathbf{w}}^t - \ddot{\mathbf{w}}^t\|$ between the global model obtained by \algdis (or \algcos) and the train-from-scratch method. 
First, we delineate the update rules for unlearning processes.

\myparatight{Estimation phase} we employ both gradient (or local model update) and Hessian information such that:    
\begin{align} 
    {\mathbf{w}}^{t+1} = {\mathbf{w}}^t - \eta \sum_{i=m+1}^n \alpha_{i} [\mathbf{H}_{i}^{t} ({\mathbf{w}}^{t} - \tilde{\mathbf{w}}^{t}) + \tilde{\mathbf{g}}_i^t], 
    \label{eq:w^{t+1} estimation} 
\end{align} 
where $\alpha_{i}$ = $\frac{|D_i|}{|D'|}$, $D' = \bigcup_{i=m+1}^n D_i$ denotes the joint dataset of the remaining  $n-m$  clients. 
For simplicity, we will adopt these notations in the subsequent analysis.

\myparatight{Exact training phase} we only use the exact gradient of each client such that:
\begin{align} 
    {\mathbf{w}}^{t+1} = \mathbf{w}^t - \eta \sum_{i=m+1}^n \alpha_{i} \mathbf{g}_i^t. 
    \label{eq:w^{t+1} update} 
\end{align}

Additionally, we denote $\mathbf{h}^t_i$ as the model update contributed by the client $i$ model in the $t$ round of the train-from-scratch strategy. The global model, updated in round $t$ is given by:
\begin{align} 
    \ddot{\mathbf{w}}^{t} = \ddot{\mathbf{w}}^{t-1} - \eta\sum_{i=m+1}^n \alpha_{i} \mathbf{h}_{i}^{t}. \label{eq:v_update} 
\end{align} 

Consider \textbf{Estimation} phase, where \(\mathbf{w}^t\) is updated according to Eq.~(\ref{eq:w^{t+1} estimation}). Based on Eqs.~(\ref{eq:w^{t+1} update}) and (\ref{eq:v_update}), we have:
\begin{align}
\|{\mathbf{w}}^{t+1}-\ddot{\mathbf{w}}^{t+1}\| & =\| ({\mathbf{w}}^{t}-\eta\sum_{i=m+1}^{n}\alpha_{i}\mathbf{g}_{i}^{t})-(\ddot{\mathbf{w}}^{t}-\eta\sum_{i=m+1}^{n}\alpha_{i}\mathbf{h}_{i}^{t})\| \nonumber \\
 & =\|{\mathbf{w}}^{t}-\ddot{\mathbf{w}}^{t}-\eta\sum_{i=m+1}^{n}\alpha_{i}(\mathbf{g}_{i}^{t}-\mathbf{h}_{i}^{t})\|,
\end{align}

Applying triangle inequality, we bound for the difference between \({\mathbf{w}}^{t+1}\) and \(\ddot{\mathbf{w}}^{t+1}\) as: 
\begin{align} & \|{\mathbf{w}}^{t+1}-\ddot{\mathbf{w}}^{t+1}\|    \nonumber \\
\leq & \underset{\mathbf{I}}{\underbrace{\| {\mathbf{w}}^{t}-\ddot{\mathbf{w}}^{t}-\eta\sum_{i=m+1}^{n}(\alpha_{i}\|\mathbf{g}_{i}^{t}-\mathbf{h}_{i}^{t}\|)\| }} \nonumber \\
 & +\underset{\mathbf{II}}{\underbrace{\| \eta\sum_{i=m+1}^{n}\alpha_{i}\left[\mathbf{g}_{i}^{t}-\mathbf{H}_{i}^{t} ({\mathbf{w}}^{t} - \tilde{\mathbf{w}}^{t}) + \tilde{\mathbf{g}}^{t}_{i}\right]\| }}.
 \label{boundI1 I2_more}
\end{align}

It remains to bound terms $\mathbf{I}$ and $\mathbf{II}$, respectively.

For term $\mathbf{I}$, we obtain:
\begin{align} \mathbf{I}^2 &= \|{\mathbf{w}}^t - \ddot{\mathbf{w}}^t\|^2 - 2\eta \langle {\mathbf{w}}_t - \ddot{\mathbf{w}}_t, \sum_{i=m+1}^n \alpha_{i} (\mathbf{g}_i^t - \mathbf{h}_i^t) \rangle \nonumber \\ &\quad + \eta^2 \|\sum_{i=m+1}^n \alpha_{i} (\mathbf{g}_i^t - \mathbf{h}_i^t)\|^2,  \end{align}
where $\left\langle \cdot,\cdot\right\rangle$ represents the inner product of two vectors. 

By applying the Cauchy-Schwarz inequality, we derive: 
\begin{align} \mathbf{I}^2 &\le (\|{\mathbf{w}}^t - \ddot{\mathbf{w}}^t\|^2 - \eta \sum_{i=m+1}^n \alpha_{i} \langle {\mathbf{w}}^t - \ddot{\mathbf{w}}^t, \mathbf{g}_i^t - \mathbf{h}_i^t \rangle) \nonumber \\ &- \eta\sum_{i=m+1}^n \alpha_{i}\langle{\mathbf{w}}^t-\ddot{\mathbf{w}}^t,\mathbf{g}_i^t-\mathbf{h}_i^t\rangle  + \eta^2 \sum_{i=m+1}^n \alpha_{i}^2 \Vert\mathbf{g}_i^t-\mathbf{h}_i^t\Vert^2.
\end{align}

Given Assumption~\ref{assumption1} and Assumption~\ref{assumption2}, we see the loss function \( \mathcal{L}_i \) is
 \( \mu \)-strong convexity and \( L \)-smoothness. Utilizing inequalities for inner products, we obtain:
\begin{align} \mathbf{I}^{2} \leq &  (\|{\mathbf{w}}^{t}-\ddot{\mathbf{w}}^{t}\|^{2}-\eta\mu\|{\mathbf{w}}^{t}-\ddot{\mathbf{w}}^{t}\|^{2}) \nonumber \\ & -\eta\ \sum_{i=m+1}^{n}(\frac{\alpha_{i}}{L}-\frac{\eta}{\alpha_{i}^{2}})\|\mathbf{g}_{i}^{t}-\mathbf{h}_{i}^{t}\|^{2}.
\end{align}

If the learning rate \( \eta \) meets the condition \( \eta \leq \frac{1}{L} \leq \frac{|D^{\prime}|}{L \max_{i=m+1}^n |D_i|} \), we obtain: 
\begin{align}
    \mathbf{I}^2 \leq (1 - \eta \mu) \|{\mathbf{w}}^t - \ddot{\mathbf{w}}^t\|^2.
\end{align}  

Furthermore, we have: 
\begin{align} \mathbf{I} \leq \sqrt{1 - \eta \mu} \|{\mathbf{w}}^t - \ddot{\mathbf{w}}^t\|. 
\label{bound of I1} 
\end{align}

For term $\mathbf{II}$, we can derive an upper bound based on Assumption~\ref{assumption2} such that: 
\begin{align} 
\mathbf{II} \le \eta M. 
\label{bound of I2} 
\end{align} 

Substituting (\ref{bound of I1}) and (\ref{bound of I2}) into (\ref{boundI1 I2_more}), we have: 
\begin{align} \|{\mathbf{w}}^{t+1} - \ddot{\mathbf{w}}^{t+1}\| \le \mathbf{I} + \mathbf{II} \le \sqrt{1-\eta\mu}\|{\mathbf{w}}^t - \ddot{\mathbf{w}}^t\| + \eta M. 
\end{align}

In the following, we consider \textbf{Exact training} case, which occurs when $t< r$ and the conditions in Eq.~(\ref{first_check}) or Eq.~(\ref{second_check}) are not satisfied. 
Similar to the previous analysis, we obtain:
\begin{align}
    \|{\mathbf{w}}^{t+1} - \ddot{\mathbf{w}}^{t+1}\| = \|{\mathbf{w}}^t - \ddot{\mathbf{w}}^t - \eta \sum_{i=m+1}^n \alpha_{i} (\mathbf{g}_i^t - \mathbf{h}_i^t)\|.
\end{align} 

It is obvious that the RHS of above equation is what we defined earlier as $\mathbf{I}$. Hence, we have:
\begin{align} \Vert{\mathbf{w}}^{t+1}-\ddot{\mathbf{w}}^{t+1}\Vert\le\sqrt{1-\eta\mu}\Vert{\mathbf{w}}^{t}-\ddot{\mathbf{w}}^t\Vert. \end{align}

Combining two cases, we have: 
\begin{align} \Vert{\mathbf{w}}^{t+1}-\ddot{\mathbf{w}}^{t+1}\Vert\le\sqrt{1-\eta\mu}\Vert{\mathbf{w}}^{t}-\ddot{\mathbf{w}}^t\Vert+\eta M. 
\end{align}

By recursion, we have: 
\begin{align}
    \|{\mathbf{w}}^t - \ddot{\mathbf{w}}^t\| \leq (\sqrt{1-\eta\mu})^t\|{\mathbf{w}}^0 - \ddot{\mathbf{w}}^0\| + \frac{1-(\sqrt{1-\eta\mu})^t}{1-\sqrt{1-\eta\mu}}\eta M, 
\end{align} 
for any $t \geq 0$, where \({\mathbf{w}}^0\) and \(\ddot{\mathbf{w}}^0\) are the initial parameter values in the FU phase and train-from-scratch, respectively. This formulation concludes that under appropriate conditions on \(\eta\), the bound converges to \(\frac{\eta M}{1 - \sqrt{1 - \eta\mu}}\) as \(t \rightarrow \infty\), thereby confirming the stability of the learning algorithm.

\end{document}